\documentclass[12pt,draftcls,onecolumn]{IEEEtran}
\usepackage{booktabs}
\usepackage{mathrsfs}
\usepackage{amsmath}
\usepackage{amssymb}
\usepackage{subfigure}
\usepackage{subfig}
\usepackage{epsfig}
\usepackage{amsfonts}
\usepackage{graphicx}
\usepackage{algorithm}
\usepackage{algorithmic}
\usepackage{cite}
\usepackage{stfloats}
\newtheorem{proposition}{Proposition}

\newtheorem{proof}{Proof}

\begin{document}
	
\title{SISO-OFDM and MISO-OFDM Counterparts for ``Wideband Waveforming for Integrated Data and Energy Transfer: Creating Extra Gain Beyond Multiple Antennas and Multiple Carriers"}

\author{Zhonglun~Wang, Jie~Hu, Kun~Yang}
\maketitle
\begin{abstract}
	In this article, we proposethe SISO-OFDM and MISO-OFDM based IDET systems, which are the counterparts of our optimal wideband waveforming strategy in \cite{WidebandWaveforming}. The fair-throughput and the sum-throughput optimisation problems are formulated. By exploiting the FP based alternating algorithm, we obtain the optimal fair-throughput and sum-throughput for both the SISO-OFDM and MISO-OFDM.
\end{abstract}

Notation: $\Vert\cdot\Vert_2$ represents the Euclidean norm operation. $\mathbb{E}\left[\cdot\right]$ represents the expection operation. $Var\left[\cdot\right]$ represents the variance operation and $Cov\left[\cdot\right]$ represents the covariance operation. $\mathcal{R}$ represents the real part of a complex variable. $`\bigotimes'$ represents the circular convolution operation. The superscript $`\dagger'$ represents the 
 conjugate transpose operation. $\mathbb{R}$ is the set of the real numbers.
\section{SISO-OFDM for IDET in Wideband Channel}	
\subsection{System Model}

We have $J$ receivers and $N$ subcarriers in the system. 
Both the transmitter and the receivers are equipped with a single antenna. 
There are a number of resolvable paths existing in the channel generated by the wideband OFDM signal. Denote the channel length of the $j$-th receiver as $L_{j}$. Then  the wideband channel vector is expressed as
\begin{align}
	\mathbf{h}_j=\left[ \underbrace{h_{j,0},\cdots,h_{j,l},\cdots,h_{j,L_{j}-1}}_{{L_j}}, \underbrace{0,\cdots,0}_{L-L_j} \right],
\end{align}
where we define the maximum channel length as $L=\max_jL_j$.


The symbol of the $j$-th receiver in the $k$-th subcarrier is denoted as $S_{j,k}$. All the receivers multiplex all the subcarriers together, while the multiplexed symbol in the $k$-th subcarrier can be expressed as
\begin{align}
	S_k=\sum_{j=1}^{J}S_{j,k}.
\end{align}
By invoking the inverse  fast Fourier transform (IFFT). Moreover,  the $n$-th symbol in the time-domain can be expressed as
\begin{align}
	s_n=\frac{1}{\sqrt{N}}\sum_{k=0}^{N-1}S_ke^{j2\pi nk/N},\forall k=0,1,\cdots,N-1.
\end{align}
In order to eliminate the inter-symbol interference (ISI) and gurantee the orthogonality of the subcarriers for the $j$-th receiver, we need to add the cyclic prefix (CP), whose period should be longer than the delay spread $\left(L_j-1\right)$ at least. Therefore, we add $\left(L-1\right)$ CP symbols to make sure we can remove the ISI for all the receivers. Therefore, the transmit symbols are sequentially expressed as
\begin{align}
	&x_0=s_0,\nonumber\\
    &\text{\ \ \ }\vdots\nonumber\\
	&x_{N-1}=s_{N-1}\nonumber\\
	&\text{\ \ \ }\vdots\nonumber\\
	&x_{N+L-2}=s_{N-2}.
\end{align}
We denote the transmit symbol vector $\mathbf{x}=\left[x_0,\cdots,x_{N+L-2}\right]$. 
Without considering the tail part of an OFDM symbol affecting the head part of the next symbol, the received symbols of the $j$-th receiver are expressed as
\begin{align}
	\mathbf{r}_{j}&=\mathbf{h}_j*\mathbf{x}
	=
		\begin{bmatrix}
			&h_{j,0} & 0   & \cdots & 0 & 0 &\cdots & 0\nonumber\\
			&\vdots&\vdots&\cdots&\vdots & 0 &\cdots& 0 \nonumber \\
			&h_{j,L-1}&h_{j,L-2}&\cdots&h_{j,0} & 0 &\cdots& 0 \nonumber\\
			&0 & h_{j,L-1} & \cdots &h_{j,1} &h_{j,0} &\cdots& 0\nonumber\\
			&\vdots&\vdots&\vdots&\vdots&\vdots&\vdots&\vdots\nonumber\\
			&0&0&\cdots&0&0&\cdots&h_{j,0}\nonumber\\
			&\vdots&\vdots&\vdots&\vdots&\vdots&\vdots&\vdots\nonumber\\
			&0&0&\cdots&0&0&\cdots&h_{j,L-1}			
		\end{bmatrix}
	\begin{bmatrix}
		&s_0\nonumber\\
		&\vdots\nonumber\\
		&s_{L-1}\nonumber\\
		&\vdots\nonumber\\
		&s_{N-1}\nonumber\\
		&s_0\nonumber\\
		&\vdots\nonumber\\
		&s_{L-2}
	\end{bmatrix}
	+\mathbf{z}_{r,j},
\tag{5}\label{eq:5}
\end{align}
where $\mathbf{z}_{r,j}$ is the circularly symmetric complex white Gaussian noise introduced by the receive antenna with a zeros mean and a variance of $\sigma_{0}^2$.
The $L$ to $(L+N-1)$ elements of the $\mathbf{r}_j$ are the circular convolution between the channel vector $\mathbf{h}_j$ and the OFDM symbol vector $\mathbf{s}$. After removing the CP,  the input signals to  the information decoder for the $j$-th receiver can be expressed as 
\begin{align}
	\mathbf{r}_{ID,j}&=\sqrt{1-\rho_j}\left[ y_{j,0},\cdots,y_{j,N-1} \right] + \sqrt{1-\rho_j}\mathbf{z}_j +\mathbf{v}_j\nonumber\\
	&=\sqrt{1-\rho_j}
	\begin{bmatrix}
		&h_{j,L-1}&h_{j,L-2}&\cdots&h_{j,0}&0&\cdots&0\nonumber\\
		&0&h_{j,L-1}&\cdots&h_{j,1}&h_{j,0}&\cdots&0\nonumber\\
		&\vdots&\vdots&\vdots&\vdots&\vdots&\vdots&\vdots\nonumber\\
		&0&0&\cdots&h_{j,L-1}&h_{j,L-2}&\cdots&0\nonumber\\
		&\vdots&\vdots&\vdots&\vdots&\vdots&\vdots&\vdots\nonumber\\
		&0&0&\cdots&0&0&\cdots&h_{j,0}
	\end{bmatrix}
\begin{bmatrix}
	&s_0\nonumber\\
	&s_1\nonumber\\
	&\vdots\nonumber\\
	&s_{N-1}
\end{bmatrix}\nonumber\\
&\text{\ \ \ }+\sqrt{1-\rho_j}\mathbf{z}_j+\mathbf{v}_j\nonumber\\
&=\sqrt{1-\rho_j}\mathbf{h}_{j}\otimes\mathbf{s}+\sqrt{1-\rho_j}\mathbf{z}_j+\mathbf{v}_j\tag{6}.
\end{align}
where $\mathbf{z}_j=\left[z_{r,j,L-1},\cdots,z_{r,j,N+L-2}\right]$  and $\mathbf{v}_j\sim\mathcal{CN}\left({0},\sigma_{cov}^2\mathbf{I}\right)$ is the passband-to-baseband conversion noise.

Then we decode the OFDM symbol by exploiting the orthogonality between the subcarriers, which has the same form as the Fast-Fourier-Transformation (FFT). The $k$-th decoding symbol for the $j$-th receiver in the frequency domain can be expressed as
\begin{small}
\begin{align}
	Y_{j,k}&= \frac{1}{\sqrt{N}}\sum_{n=0}^{N-1}\mathbf{r}_{ID,j}  \nonumber\\
	&=\sqrt{N}\sqrt{1-\rho_j}H_{j,k}S_k+\sqrt{1-\rho_j}Z_{j,k}+V_{j,k},
\end{align}
\end{small}
where $H_{j,k}$ and $S_k$ are the FFT of the $h_{j,n}$ and $s_n$, which can be expressed as
\begin{align}
	H_{j,k}&=\frac{1}{\sqrt{N}}\sum_{n=0}^{L-1}{h_{j,n}}e^{-j2\pi nk/N},\\
	S_k&=\frac{1}{\sqrt{N}}\sum_{n=0}^{N-1}{s_{n}}e^{-j2\pi nk/N},\\
	Z_{j,k}&=\frac{1}{\sqrt{N}}\sum_{n=0}^{N-1}{z_{j,k}}e^{-j2\pi nk/N},\\
	V_{j,k}&=\frac{1}{\sqrt{N}}\sum_{n=0}^{N-1}{v_{j,k}}e^{-j2\pi nk/N}.
\end{align}
The noise $Z_{j,k}$ and $V_{j,k}$ have the same distribution with $z_{j,k}$ and $v_{j,k}$, respectively \cite{Fundementals_of_communication}.

Then the signal-to-interference-plus-noise ratio (SINR) $\gamma_{j,k}$ can be expressed as 
\begin{align}
	\gamma_{j,k} \left( \mathbf{p},\rho_{j} \right) &=\frac {(1-\rho_j)N\left| H_{j,k}\right|^2\mathbb{E}\left[S_{j,k}S_{j,k}^{\dagger}\right]}   {(1-\rho_j)\sum_{j'\neq j}N\left| H_{j',k}\right|^2\mathbb{E}\left[S_{j',k}S_{j',k}^{\dagger}\right]  + (1-\rho_j)\sigma_{0}^2+\sigma_{cov}^2 }\nonumber\\
	&=\frac {(1-\rho_j)N\left| H_{j,k}\right|^2p_{j,k}}   {(1-\rho_j)\sum_{j'\neq j}N\left| H_{j',k}\right|^2p_{j',k}  + (1-\rho_j)\sigma_{0}^2 + \sigma_{cov}^2 },\label{eq:eq11}
\end{align}
where $\mathbf{p}=\left[ [p_{1,0},\cdots,p_{1,N-1}]^T,\cdots,[p_{J,0},\cdots,p_{J,N-1}]^T \right]$.
The throughput $R_j\left( \mathbf{p} , \rho_{j} \right)$ is  expressed as
\begin{align}
	R_j\left( \mathbf{p} , \rho_{j} \right)= \frac{B}{N}\cdot\frac{N}{N+L-1}  \sum_{k=0}^{N-1}\log_2\left( 1+\gamma_{j,k}\left( \mathbf{p} , \rho_{j} \right) \right),\text{\ bit/s}.\label{eq:eq12}
\end{align}
Let the average transmit power of each symbol $s_n$ be $P_{tx}^{OFDM}$. According to the Perseval's theorem \cite{DSP}, we can obtain
\begin{align}
			\frac{1}{{N}}\sum_{k=0}^{N-1}\left| S_k\right|^2=P_{tx}^{OFDM},\\
			\frac{1}{{N}}\sum_{n=0}^{N-1}\left| s_n\right|^2=P_{tx}^{OFDM},\\
			\frac{1}{{N+L-1}}\sum_{n=0}^{N+L-2}\left| x_n\right|^2=P_{tx}^{OFDM}.
\end{align}
Therefore, the power constraints for the OFDM based IDET system can be expressed as
\begin{align}
	0\leq \sum_{k=0}^{M-1}\mathbb{E}\left[ S_kS_k^\dagger \right] =\sum_{k=0}^{M-1}\sum_{j=1}^{J}\mathbb{E}\left[ S_{j,k}S_{j,k}^\dagger \right]=\sum_{k=0}^{N-1}\sum_{j=1}^{J}p_{j,k}\leq NP_{tx}^{OFDM}.
\end{align}

	The  RF signal for the energy harvester of the $j$-th receiver can be expressed as
	\begin{align}
		\mathbf{r}_{EH,j}=\sqrt{\rho_j}\left[ \mathbf{r}_{CP,j}^T , \mathbf{r}_{OS,j}^T \right]^T,
	\end{align}
where $\mathbf{r}_{CP,j}$ is the received CP symbol with considering the tail part of an OFDM symbol affecting the  head part of the next symbol and $r_{OS,j}$ is the receive OFDM symbol.
The power for the energy harvester of the $j$-th receiver is $P_{EH,j}=\frac{{\rho_j}}{N+L-1}\left(E_{CP,j}+E_{OS,j}\right)$, where $E_{CP,j}$ is the energy of the received CP symbols and $E_{OS,j}$ is the energy of the received OFDM symbols.    The received CP symbols can be expressed
	$$ \left\{
\begin{aligned}
	r_{CP,j,0} & = h_{j,0} s_{0} + h_{j,L-1} s_{0}' + h_{j,L-2} s_{1}' + \cdots + h_{j,1} s_{L-2}' , \\
	r_{CP,j,1} & = h_{j,1} s_{0} + h_{j,2} s_{1} + h_{j,L-1} s_{1}' + \cdots + h_{j,2} s_{L-2}',\\
	&\text{\ \ \ \ \ \ \ \ \ \ \ \ \ \ \ \ \ \  \ }\vdots\\
	r_{CP,j,n} & = h_{j,n} s_0 + \cdots + h_{j,0} s_{n} + h_{j,L-1} s_{n}' + \cdots + h_{j,n+1} s_{L-2},\\
	&\text{\ \ \ \ \ \ \ \ \ \ \ \ \ \ \ \ \ \  \ }\vdots\\
	r_{CP,j,L-2} & = h_{j,L-2} s_0 +h_{j,L-3} s_{1} +  \cdots + h_{j,0} s_{L-2} + h_{j,L-1} s_{L-2}'.
\end{aligned}
\right.
$$
\begin{proposition}
	The time domain symbol $s_n \sim \mathcal{CN}\left( 0 , P_{tx}^{OFDM} \right) $, for $\forall n=0,\cdots,N-1$.
\end{proposition}
\begin{proof}
	See Appendix A.
\end{proof}
\begin{proposition}
	The symbol $s_{n_1}$ in the current OFDM period and the symbol $s_{n_2}'$ in the former OFDM period are independent each other, for $\forall n_1,n_2$.
\end{proposition}
\begin{proof}
	See Appendix B.
\end{proof}
According to Theorem 1 and Theorem 2, the energy carried by the received CP can be expressed as
\begin{align}
	E_{CP,j} & = \mathbb{E}\left[ \sum_{n=0}^{L-2} \left| r_{CP,j,n} \right| \right]\nonumber\\
	& = \mathbb{E}\left[ \sum_{n=0}^{L-2} \left| \left( \sum_{i=0}^{n}h_{j,i}s_{n-i} + \sum_{i=n}^{L-2}h_{j,L-1+n-i}s_{i}' + z_{r,j,n} \right)\right|^2 \right]\nonumber\\
	& = \left(L-1\right)P_{tx}^{OFDM}\sum_{n=0}^{L-1}\left|h_{j,n}\right|^2 + \left(L-1\right) \sigma_{0}^2 .
\end{align}
The energy carried by the received OFDM symbols $E_{OS,j}$ can be expressed as
\begin{align}
	E_{OS,j}&=\sum_{n=L-1}^{N+L-2}\left| r_{j,n} \right|^2\nonumber\\
	&=N\sum_{k=0}^{N-1}\left| H_{j,k} \right|^2\left| S_k \right|^2 + \sum_{k=0}^{N-1} \left| Z_{j,k} \right|^2 \nonumber\\
	&=N\sum_{k=0}^{N-1}\sum_{j=1}^{J}p_{j,k}\left| H_{j,k} \right|^2 + N\sigma_{0}^2 .
\end{align}
Then the receive power	for the energy harvester of the $j$-th receiver can be reformulated as
\begin{align}
	P_{EH,j} \left( \mathbf{p} , \rho_{j} \right) &=\frac{\rho_j}{N+L-1} \left(E_{CP,j} +E_{OS,j} \right).\label{eq:eq20}
\end{align}
The RF-DC conversion is determined by the charateristics of the diode. Given an input RF power of $ P_{EH}(\mathbf{p},\rho_{j}) $, the output DC is expressed as
\begin{align}
	i_j^{SO}(\mathbf{p},\rho_{j}) \approx k_0 + k_1 P_{EH}(\mathbf{p},\rho_{j}) + k_2 P_{EH}^2(\mathbf{p},\rho_{j}), \label{eq:DC}
\end{align}
where the circuit related constants satisfies $k_0 \approx 0$, $k_1 \geq 0$ and $k_2 \geq 0$, respectively \cite{DC}.

\subsection{Fair-throughput Maximisation}

Our fair-throughput optimisition problem for the SISO-OFDM based IDET system can be formulated as
\begin{align}
	&\text{(P1):}\max_{ \left\{p_{j,k}\right\},\left\{\rho_j\right\} }R_{fair}^{SO}\label{Problem 1}\\
	&\text{s. t.    } R_j \left( \mathbf{p} , \rho_{j} \right) \geq R_{fair}^{SO},\forall j=1,\cdots,J,\tag{\ref{Problem 1}a}\label{{Problem1}a}\\
	&\text{         }\ \ \ \ \  i_j^{SO}(\mathbf{p},\rho_{j}) \geq \mathcal{I}_{j}^{SO},\tag{\ref{Problem 1}b}\label{{Problem1}b}\nonumber\\
	&\text{         }\ \ \ \ \  0\leq\sum_{k=0}^{N-1}\sum_{j=1}^{J}p_{j,k}\leq NP_{tx}^{OFDM},\tag{\ref{Problem 1}c}\label{{Problem1}c}\\
	&\text{         }\ \ \ \ \  0\leq\rho_j\leq1,\tag{\ref{Problem 1}d}\label{{Problem1}d}
\end{align}
where $R_{fair}^{SO}$ is the fair-throughput and $\mathcal{I}_j^{SO}$ is the DC requirement for the $j$-th receiver.
Observe from Eqs. \eqref{eq:eq11}-\eqref{eq:eq12} and Eq. \eqref{eq:eq20}, the expression of  $R_j\left( \mathbf{p} , \rho_{j} \right)$  is non-convex. By introducing a series of  auxiliary variables $\left\{ \psi_{j,k} \right\}$, we can reformulated the SINR $\gamma_{j,k}\left( \mathbf{p} , \rho_{j} \right)$ as 
\begin{align}
	\widetilde{\gamma}_{j,k}\left( \mathbf{p} , \rho_{j} , \boldsymbol{\psi}_j \right)&=  2\psi_{j,k} \left| H_{j,k} \right| \sqrt{ (1-\rho_j) N p_{j,k}  }     \nonumber\\
	&\text{\ \ \ }-\psi_{j,k}^2 \left( (1-\rho_j)\sum_{j'\neq j}N\left| H_{j',k}\right|^2p_{j',k}  + (1-\rho_j)\sigma_{0}^2 + \sigma_{cov}^2 \right) , \label{eq:eq22}
\end{align}
where $\boldsymbol{\psi}_j=\left[ \psi_{j,0},\cdots,\psi_{j,N-1} \right]^T$ \cite{FP}.
Moreover, the throughput for the $j$-th receiver can be reformulated as
\begin{align}
	\widetilde{R}_j\left( \mathbf{p} , \rho_{j} , \boldsymbol{\psi}_{j} \right)=\frac{B}{N+L-1}\sum_{k=0}^{N-1}\log_2\left( 1 + \widetilde{\gamma}_{j,k}\left( \mathbf{p} , \rho_{j} , \boldsymbol{\psi}_{j} \right) \right).
\end{align}
The constraint \eqref{{Problem1}b} is non-convex due to the quadratic structure of Eq. \eqref{eq:DC} with respect to $P_{EH,j}(\mathbf{p},\rho_{j})$. By adopting the quadratic transformation, the constraint \eqref{{Problem1}b} is reformulated as
\begin{align}
	P_{EH,j}(\mathbf{p},\rho_{j})\geq \frac{-k_1+\sqrt{k_1^2+4k_2\mathcal{I}_j^{SO}}}{2k_2},\forall j=1,\cdots,N.\label{powerconstraint1}
\end{align}

Then  (P1) can be reformulated as
\begin{align}
	&\text{(P2):}\max_{ \left\{p_{j,k}\right\},\left\{\rho_j\right\},\left\{\psi_{j,k}\right\} }\widetilde{R}_{fair}^{SO}\label{Problem 2}\\
	&\text{s. t.    } \widetilde{R}_j\left( \mathbf{p} , \rho_{j} , \boldsymbol{\psi}_{j} \right)\geq \widetilde{R}_{fair}^{SO},\forall j=1,\cdots,J,\tag{\ref{Problem 2}a}\label{{Problem2}a}\\
	&\text{         }\ \ \ \ \  \psi_{j,k}\in\mathbb{R},\forall j,k,\tag{\ref{Problem 2}b}\label{{Problem2}b}\nonumber\\
	&\text{         }\ \ \ \ \  \eqref{powerconstraint1},\eqref{{Problem1}c},
\end{align}
By alternatively optimising the $\left\{ p_{j,k} \right\}$, $\left\{ \rho_j \right\}$ and $\left\{ \psi_{j,k} \right\}$, we can obtain the optimal solution for (P2), which is also optimal  for (P1). First, we optimise $\left\{p_{j,k}\right\}$ by fixing $\left\{\rho_j\right\}$ and $\psi_{j,k}$, which can be reformulated as 
\begin{align}
	&\text{(P2-1):}\max_{ \left\{p_{j,k}\right\} }\widetilde{R}_{fair}^{SO}\label{Problem 2-1}\\
	&\text{s. t.    } \eqref{{Problem2}a}, \eqref{powerconstraint1},\eqref{{Problem1}c}.\nonumber
\end{align}
  By fixing $\left\{p_{j,k}\right\}$ and $\left\{\psi_{j,k}\right\}$, the optimisition problem (P2) can be reformulated as
  \begin{align}
  	&\text{(P2-2):}\max_{ \left\{\rho_j\right\} }\widetilde{R}_{fair}^{SO}\label{Problem 2-2}\\
  	&\text{s. t.    }  \eqref{{Problem2}a},
     \eqref{powerconstraint1},\eqref{{Problem1}d}.\nonumber
  \end{align}
	Then we fix $\left\{p_{j,k}\right\}$ and $\left\{ \rho_j \right\}$, we can obtain the optimal $\left\{\psi_{j,k}\right\}$ expressed as
	\begin{align}
		\psi_{j,k}^*&=\frac{\Vert H_{j,k} \Vert_2 \sqrt{ (1-\rho_j) N p_{j,k}  }}{ (1-\rho_j)\sum_{j'\neq j}N\Vert H_{j',k}\Vert_2^2p_{j',k}  + (1-\rho_j)\sigma_{0}^2 + \sigma_{cov}^2 }.\label{eq:eq27}
	\end{align}
We summarize the process in Algorithm 1. The convergence of the Algorithm 1 is proved  in Appexdix C. Since (P1) is a FP pproblem with the convex constraints, we can obtain the global optimal solution \cite{FP}.
\begin{algorithm}[!t]
	\linespread{1}\selectfont
	\caption{FP based alternating optimisation for solving (P2).}
	\label{alg:alg1}
	\footnotesize
	\begin{algorithmic}[1]
		\REQUIRE\
		Channel fading coefficient of $\mathbf{h}_j$; Transmit power of the BS $P_{tx}$;  Bandwidth of the signal $B$; Sampling factor $\mathcal{D}$; DC requirement $\mathcal{I}_j^{SO}$; Error tolerance $\epsilon$.
		\ENSURE\
		Optimal fair-throughput $R_{fair}^{OFDM*}$.
		\STATE Initialize $\{\rho\}\leftarrow {1}$,  $\{\psi_{j,k}\}\leftarrow {1}$, $\widetilde{R}_{fair}^{SO1}\leftarrow \epsilon$, $\widetilde{R}_{fair}^{SO2}\leftarrow -\epsilon$;
		\WHILE{$|\widetilde{R}_{fair}^{SO1}-\widetilde{R}_{fair}^{SO2}|\geq\epsilon$}
		\STATE $\widetilde{R}_{fair}^{SO2}\leftarrow\widetilde{R}_{fair}^{SO1}$;
		\STATE Update $\left\{ p_{j,k}^* \right\}$ by solving (P2-1);
		\STATE Update $\left\{\rho_j^*\right\}$ by solving (P2-2);
		\STATE Update $\{{\psi}_{j,k}^*\}$ by Eq. \eqref{eq:eq27};
		\STATE Update $\widetilde{R}_{fair}^{SO1}\leftarrow \widetilde{R}_{fair}^{SO}$;
		\ENDWHILE
		\STATE Update $R_{fair}^{OFDM*} \leftarrow \widetilde{R}_{fair}^{OFDM}$
		\RETURN  $R_{fair}^{OFDM*}$.
	\end{algorithmic}
\end{algorithm}
\subsection{Sum-Throughput Maximization}
The sum-throughput maximisation problem is formulated as
\begin{align}
&\text{(P3):}\max_{ \left\{p_{j,k}\right\},\left\{\rho_j\right\} }  R_{sum}^{SO}=\sum_{j=1}^{J} R_{j}\left( \mathbf{p} , \rho_{j} \right)     \label{Problem 3}\\
&\text{s. t.    } \eqref{powerconstraint1},\eqref{{Problem1}c},\eqref{{Problem1}d}.\nonumber
\end{align}
Similarily, by  introducing the auxilary variables $\{ \psi_{j,k} \}$ and adopting the FP technology, (P3) is reformulated as
\begin{align}
&\text{(P4):}\max_{ \left\{p_{j,k}\right\},\left\{\rho_j\right\}, \{ \psi_{j,k} \} }  \widetilde{R}_{sum}^{SO} = \sum_{j=1}^{J} \widetilde{R}_{j} \left( \mathbf{p} , \rho_{j} , \boldsymbol{\psi}_{j} \right)    \label{Problem 4}\\
&\text{s. t.    } \eqref{powerconstraint1},\eqref{{Problem1}c},\eqref{{Problem1}d}.\nonumber
\end{align}
(P4) can be solved by alternatively optimising $\{p_{j,k}\}$, $\{\rho_{j}\}$ and $\{ \psi_{j,k} \}$. First, given fixed $\{ \rho_{j} \}$ and $ \{ \psi_{j,k} \} $, (P4) is reformulated as
\begin{align}
&\text{(P4-1):}\max_{ \left\{p_{j,k}\right\} } \widetilde{R}_{sum}^{SO} = \sum_{j=1}^{J} \widetilde{R}_{j}  \left( \mathbf{p} , \rho_{j} , \boldsymbol{\psi}_{j} \right)   \label{Problem 4-1}\\
&\text{s. t.    } \eqref{powerconstraint1},\eqref{{Problem1}c}.\nonumber
\end{align}
By giving fixed $\{ p_{j,k} \}$ and $ \{ \psi_{j,k} \} $, (P4) is reformulated as
\begin{align}
&\text{(P4-2):}\max_{ \left\{p_{j,k}\right\},\left\{\rho_j\right\}, \{ \psi_{j,k} \} } \widetilde{R}_{sum}^{SO} = \sum_{j=1}^{J} \widetilde{R}_{j} \left( \mathbf{p} , \rho_{j} , \boldsymbol{\psi}_{j} \right)    \label{Problem 4-2}\\
&\text{s. t.    } \eqref{powerconstraint1},\eqref{{Problem1}d}.\nonumber
\end{align}
The sub-problems (P4-1) and (P4-2) are convex, we can solve it by any convex optimisation tool. Finally, given fixed $\{ p_{j,k} \}$ and $\{ \rho_{j} \}$, (P4) is reformulated as
\begin{align}
&\text{(P4-3):}\max_{\{ \psi_{j,k} \} } \widetilde{R}_{sum}^{SO} = \sum_{j=1}^{J} \widetilde{R}_{j}\left( \mathbf{p} , \rho_{j} , \boldsymbol{\psi}_{j} \right).     \label{Problem 4-3}
\end{align}
The unconstrained optimisation problem (P4-3) can be solved by letting $ \partial \widetilde{R}_{sum}^{SO} /\partial \psi_{j,k} = 0 $, where the optimal $\{ \psi_{j,k}^* \}$ is expressed as \eqref{eq:eq27}.
The FP based alternating algorithm for solving (P3) is detailed in Algorithm 2. The proof of the convergence for Algorithm 2 is similar to Algorithm 1. Since (P3) is a FP based sum-throughput maximisation problem with the convex constraints, we can obtain the local optimal solution \cite{FP}.

\begin{algorithm}[!t]
	\linespread{1}\selectfont
	\caption{FP based alternating optimisation for solving (P3).}
	\label{alg:alg2}
	\footnotesize
	\begin{algorithmic}[1]
		\REQUIRE\
		Channel fading coefficient of $\mathbf{h}_j$; Transmit power of the BS $P_{tx}$;  Bandwidth of the signal $B$; Sampling factor $\mathcal{D}$; DC requirement $\mathcal{I}_{j}^{SO}$; Error tolerance $\epsilon$.
		\ENSURE\
		Optimal sum-throughput $R_{sum}^{SO*}$.
		\STATE Initialize $\{\rho\}\leftarrow {1}$,  $\{\psi_{j,k}\}\leftarrow {1}$, $\widetilde{R}_{sum}^{SO1}\leftarrow \epsilon$, $\widetilde{R}_{sum}^{SO2}\leftarrow -\epsilon$;
		\WHILE{$|\widetilde{R}_{sum}^{SO1}-\widetilde{R}_{sum}^{SO2}|\geq\epsilon$}
		\STATE $\widetilde{R}_{sum}^{SO2}\leftarrow\widetilde{R}_{sum}^{SO1}$;
		\STATE Update $\left\{ p_{j,k}^* \right\}$ by solving (P4-1);
		\STATE Update $\left\{\rho_j^*\right\}$ by solving (P4-2);
		\STATE Update $\{{\psi}_{j,k}^*\}$ by Eq. \eqref{eq:eq27};
		\STATE Update $\widetilde{R}_{sum}^{SO1}\leftarrow \widetilde{R}_{sum}^{SO}$;
		\ENDWHILE
		\STATE Update $ {R}_{sum}^{SO*} \leftarrow \widetilde{R}_{sum}^{SO}  $
		\RETURN $R_{sum}^{SO*}$.
	\end{algorithmic}
\end{algorithm}

\subsection{Transmit Power Parameter}

The time-domain OFDM symbols and the direct-sequence wideband (DSW) up-sampling symbols are displayed in  Fig. 1(a) and Fig. 1(b), respectively. We transmit $\left(N+L-1\right)$ non-zero symbols with the average power $\overline{P}$ in the respective period for fairness. The transmit periods of the OFDM and the DSW are $T_{OFDM}=\left(N+L-1\right)T_s$ and $T_{DSW}=\left(N+L-1\right)\mathcal{D}T_s$, where $T_s=1/B$ is the sampling period and $B$ is the bandwidth both of the OFDM signals and DSW signals. The average transmit power for the DSW  can be expressed as  
	\begin{align}
		\overline{P} & = \frac{1}{\left(N+L-1\right)\mathcal{D}}\sum_{n=0}^{N+L-2}\mathbb{E}\left[\left| s_{DSW,n\mathcal{D}}^{\mathcal{D}} \right|^2\right].
	\end{align}
\begin{figure}[!t]
	\centering
	\includegraphics [width=1\linewidth] {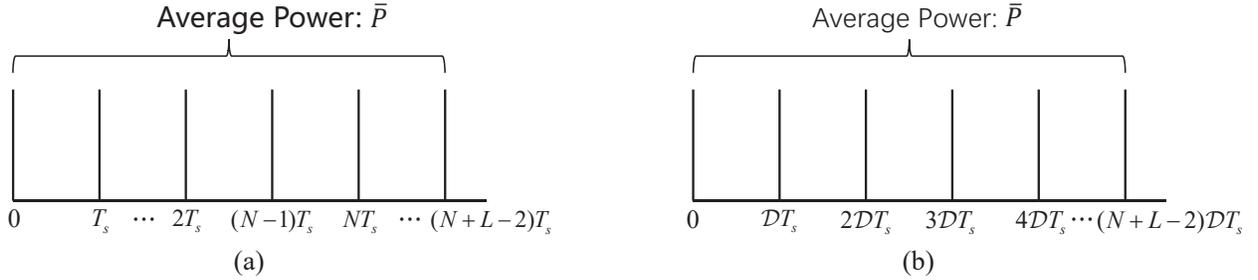}
	\caption{Transmit symbol: (a) OFDM symbol $s_n$; (b) DSW up-sampling symbol $s_{DSW,n}^{\mathcal{D}}$.}
	\setlength{\belowcaptionskip}{-10pt}
	\label{fig:fig1}
\end{figure}
Assume $s_{DSW,n\mathcal{D}}^{\mathcal{D}}\sim\mathcal{CN}\left(0,P_{tx}^{DSW}\right)$, for $\forall n=0,\cdots,N+L-2$. We can obtain the average power of each symbol $P_{tx}^{DSW}$  as
\begin{align}
	P_{tx}^{DSW}&=\mathbb{E} \left[ \left| s_{DSW,n\mathcal{D}}^{\mathcal{D}} \right|^2 \right]=\mathcal{D}\overline{P.}
\end{align}
Similarly, the average transmit power for the OFDM can be expressed as
\begin{align}
	\overline{P}=\frac{1}{N+L-1}\sum_{n=0}^{N+L-2}\mathbb{E}\left[ \left| s_{n} \right|^2 \right]. 
\end{align}
As mentioning in  $s_n\sim \mathcal{CN}\left(0,P_{tx}^{OFDM}\right)$, the average transmit power of each OFDM symbol $P_{tx}^{OFDM}$ can be expressed as
\begin{align}
	P_{tx}^{OFDM}=\mathbb{E}\left[ \left| s_n \right|^2 \right] =\overline{P}. 
\end{align}
Finally, for comparing the DSW based IDET scheme and the OFDM based IDET scheme fairly, we must let the ratio between the average power of the OFDM symbols and the average power of the non-zero DSW symbols satisfy $P_{tx}^{DSW}/P_{tx}^{OFDM}=\mathcal{D}$.

\section{MISO OFDM}
\subsection{System Model}
We have an OFDM transmitter with $M$ antennas and $J$ single antenna receivers. The number of the OFDM subcarriers is denoted as $N$. Assume the wideband resolvable channel from the $m$-th transmit antenna to the $j$-th receiver as $\mathbf{h}_{j,m}=\left[ h_{j,m,0}\cdots,h_{j,m,L} \right]$, where  the channel length for all links is $L\leq N$. The vector $\mathbf{H}_{j,k}=[H_{j,1,k},\cdots,H_{j,M,k}]$ is defined as the FFT channel in the $k$-th subcarrier between the transmitter and the $j$-th receiver, where the $m$-th element $H_{j,m,k}$ is expressed as
\begin{align}
	H_{j,m,k}=\frac{1}{\sqrt{N}}\sum_{n=0}^{L-1}h_{j,m,n}e^{-j2\pi nk/N},\forall k=0,\cdots,N-1.
\end{align}
The $k$-th symbol of the $j$-th receiver in the $m$-th antenna is $S_{j,m,k}$, which satisfies  $S_{j,1,k}=\cdots=S_{j,M,k}$ and $S_{j,m,k}\sim \mathcal{CN}(0,1)$. Moreover, the optimal beamforming vector of the $j$-th receiver for the $k$-th subcarrier is defined as $\mathbf{W}_{j,k}=[W_{j,1,k},\cdots,W_{j,M,k}]^T$, which satisfies $\Vert \mathbf{W}_{j,k} \Vert_2^2\leq p_{j,k}$. The $k$-th symbol of the $m$-th antenna $S_{m,k}$ is expressed as
\begin{align}
	S_{m,k}=\sum_{j=1}^{J}W_{j,m,k}S_{j,m,k}.
\end{align}
Then time-domain symbol  $s_{m,n}$ is  expressed as 
\begin{align}
	s_{m,n}=\frac{1}{\sqrt{N}}\sum_{k=0}^{N-1}S_{m,k}e^{j2\pi nk/N}, \forall n=0,\cdots,N-1.
\end{align}
In order to eliminate the inter-symbol interference (ISI) and gurantee the orthogonality of the subcarriers for the $j$-th receiver, we also need to add the CP whose period is longger the delay spread $\left(L_j-1\right)$ at least. Therefore, we add $\left(L-1\right)$ CP symbols to make sure we can remove the ISI for all the receivers.  The transmit symbols are then  expressed as
	$$ \left\{
	\begin{small}
\begin{aligned}
x_{m,0} & =s_{m,0},\\
&\vdots \\
x_{m,N-1}&=s_{m,N-1},\\
x_{m,N}&=s_{m,0},\\
\vdots\\
x_{m,N+L-2}&=s_{m,L-2}.
\end{aligned}
\end{small}
\right.
$$
The received symbol of the $j$-th receiver $\mathbf{r}_j$ is expressed as
\begin{small}
 \begin{align}
 	\mathbf{r}_j&=\sum_{m=1}^{M}\mathbf{h}_{j,m}*\mathbf{x}_m+\mathbf{z}_j\nonumber\\
 	&=\sum_{m=1}^{M}	
\begin{bmatrix}
&h_{j,m,0} & 0   & \cdots & 0 & 0 &\cdots & 0\nonumber\\
&\vdots&\vdots&\cdots&\vdots & \vdots  &\cdots& \vdots  \nonumber \\
&h_{j,m,L-1}&h_{j,m,L-2}&\cdots&h_{j,m,0} & 0 &\cdots& 0 \nonumber\\
&0 & h_{j,m,L-1} & \cdots &h_{j,m,1} &h_{j,m,0} &\cdots& 0\nonumber\\
&\vdots&\vdots&\vdots&\vdots&\vdots&\vdots&\vdots\nonumber\\
&0&0&\cdots&0&0&\cdots&h_{j,m,0}\nonumber\\
&\vdots&\vdots&\vdots&\vdots&\vdots&\vdots&\vdots\nonumber\\
&0&0&\cdots&0&0&\cdots&h_{j,m,L-1}			
\end{bmatrix}
\begin{bmatrix}
&s_{m,0}\nonumber\\
&\vdots\nonumber\\
&s_{m,L-1}\nonumber\\
&\vdots\nonumber\\
&s_{m,N-1}\nonumber\\
&s_{m,0}\nonumber\\
&\vdots\nonumber\\
&s_{m,L-2}
\end{bmatrix}
+
\begin{bmatrix}
&z_{j,0}\nonumber\\
&\vdots\nonumber\\
&z_{j,L-1}\nonumber\\
&\vdots\nonumber\\
&z_{j,N-1}\nonumber\\
&z_{j,N}\nonumber\\
&\vdots\nonumber\\
&z_{j,N+L-2}
\end{bmatrix},
 \end{align}
\end{small}
where $\mathbf{z}_j\sim \mathcal{CN}(0,\sigma_{0}^2\mathbf{I})$ is the antenna noise.  By removing the CP in the receivers, the information symbol of the $j$-th receiver  $\mathbf{r}_{ID,j}$ is expressed as
\begin{small}
\begin{align}
	\mathbf{r}_{ID,j}=&\sqrt{1-\rho_j}\sum_{m=1}^M\mathbf{h}_{j,m}\otimes \mathbf{s}_m+\sqrt{1-\rho_j}\widetilde{\mathbf{z}}_j+\widetilde{\mathbf{v}}_j\nonumber\\
	=&\sqrt{1-\rho_j}\sum_{m=1}^M	
	\begin{bmatrix}
	&h_{j,m,L-1}&h_{j,m,L-2}&\cdots&h_{j,m,0}&0&\cdots&0\nonumber\\
	&0&h_{j,m,L-1}&\cdots&h_{j,m,1}&h_{j,m,0}&\cdots&0\nonumber\\
	&\vdots&\vdots&\vdots&\vdots&\vdots&\vdots&\vdots\nonumber\\
	&0&0&\cdots&h_{j,m,L-1}&h_{j,m,L-2}&\cdots&0\nonumber\\
	&\vdots&\vdots&\vdots&\vdots&\vdots&\vdots&\vdots\nonumber\\
	&0&0&\cdots&0&0&\cdots&h_{j,m,0}
	\end{bmatrix}
	\begin{bmatrix}
	&s_{m,0}\nonumber\\
	&s_{m,1}\nonumber\\
	&\vdots\nonumber\\
	&s_{m,N-1}
	\end{bmatrix}\nonumber\\
	&+\sqrt{1-\rho_j}\widetilde{\mathbf{z}}_j+\widetilde{\mathbf{v}}_j,
\end{align} 
\end{small}
where $\widetilde{\mathbf{z}}_j$ is the antenna noise and $\widetilde{\mathbf{v}}_j\sim \mathcal{CN}(0,\sigma_{cov}^2\mathbf{I})$ is the RF-to-baseband conversion noise. Then we decode the OFDM symbols by exploiting the orthogonality between the subcarriers, which has the same form with the DFT. The $k$-th decoding symbol of the $j$-th receiver $Y_{j,k}$ is expressed as 
\begin{align}
	Y_{j,k}&=\frac{1}{\sqrt{N}}\sum_{n=0}^{N-1}r_{ID,j,n}e^{-j2\pi nk/N}\nonumber\\
	&=\sqrt{N}\sqrt{1-\rho_j}\sum_{m=1}^MH_{j,m,k}S_{m,k}+\sqrt{1-\rho_j}Z_{j,k}+V_{j,k}.\nonumber\\
	&=\sqrt{N}\sqrt{1-\rho_j}\sum_{m=1}^M\sum_{j=1}^{J}H_{j,m,k}S_{j,m,k}+\sqrt{1-\rho_j}Z_{j,k}+V_{j,k}.
\end{align}
Moreover, the average SINR $\gamma_{j,k}$ is expressed as
\begin{align}
	\gamma_{j,k} \left( \mathbf{W} , \rho_{j} \right) &=\frac{N(1-\rho_{j})\mathbf{H}_{j,k}\mathbf{W}_{j,k}\mathbf{W}_{j,k}^\dagger\mathbf{H}_{j,k}^\dagger }{ N(1-\rho_{j}) \sum_{j'\neq j} \mathbf{H}_{j,k}\mathbf{W}_{j',k}\mathbf{W}_{j',k}^\dagger\mathbf{H}_{j,k}^\dagger  +(1-\rho_{j})\sigma_{0}^{2} + \sigma_{cov}^{2}  }
\end{align}
where $\mathbf{W}=\left[ [\mathbf{W}_{1,0},\cdots,\mathbf{W}_{1,N-1}]^T,\cdots,[\mathbf{W}_{J,0},\cdots,\mathbf{W}_{J,N-1}]^T \right]$.
Finally, the throughput of the $j$-th receiver $R_j$ is expressed as
\begin{align}
	R_j\left( \mathbf{W} , \rho_{j} \right) =\frac{B}{N}\frac{N}{N+L-1}\sum_{k=0}^{N-1}\log_2(1+\gamma_{j,k}\left( \mathbf{W} , \rho_{j} \right)),\text{bit/s}.
\end{align}
Considering the tail part of an OFDM symbol affecting the the head part of the next symbol, the CP $\{r_{CP,j,m,k}\}$ can be expressed as
	$$ \left\{
\begin{aligned}
r_{CP,j,0} & = \sum_{m=1}^{M}h_{j,m,0} s_{m,0} + \sum_{m=1}^{M}h_{j,m,L-1} s_{m,0}' + \sum_{m=1}^{M}h_{j,m,L-2} s_{m,1}' + \cdots + \sum_{m=1}^{M}h_{j,m,1} s_{m,L-2}' , \\
r_{CP,j,1} & = \sum_{m=1}^{M}h_{j,m,1} s_{m,0} + \sum_{m=1}^{M}h_{j,m,2} s_{m,1} + \sum_{m=1}^{M}h_{j,m,L-1} s_{m,1}' + \cdots + \sum_{m=1}^{M}h_{j,m,2} s_{m,L-2}',\\
&\text{\ \ \ \ \ \ \ \ \ \ \ \ \ \ \ \ \ \  \ }\vdots\\
r_{CP,j,n} & = \sum_{m=1}^{M}h_{j,m,n} s_{m,0}  + \cdots + \sum_{m=1}^{M}h_{j,m,0} s_{m,n} + \sum_{m=1}^{M}h_{j,m,L-1} s_{m,n}' + \cdots + \sum_{m=1}^{M}h_{j,m,n+1} s'_{m,L-2},\\
&\text{\ \ \ \ \ \ \ \ \ \ \ \ \ \ \ \ \ \  \ }\vdots\\
r_{CP,j,L-2} & = \sum_{m=1}^{M}h_{j,m,L-2} s_{m,0} +\sum_{m=1}^{M}h_{j,m,L-3} s_{m,1} +  \cdots + \sum_{m=1}^{M}h_{j,m,0} s_{m,L-2} + \sum_{m=1}^{M}h_{j,m,L-1} s_{m,L-2}'.
\end{aligned}
\right.
$$ 
The energy  $E_{CP,j}$  harvested by  the $j$-th receiver  from CP is expressed as
\begin{small}
\begin{align}
	E_{CP,j}=&\mathbb{E}\left[ \sum_{n=0}^{L-2}\left| r_{CP,j,n} \right|^2        \right]\nonumber\\
	=&\mathbb{E}\left[ \sum_{n=0}^{L-2}\left| \sum_{i=0}^n\sum_{m=1}^Mh_{j,m,i}s_{m,n-i} +\sum_{i=n}^{L-2}\sum_{m=1}^Mh_{j,m,L-1+n-i}s'_{m,i} \right|^2        \right]+\sum_{n=0}^{L-2}\mathbb{E}\left[\left|z_{j,n}\right|^2\right]\nonumber\\
	=&\sum_{n=0}^{L-2} \sum_{i=0}^{n}\mathbb{E}\left[\left|\sum_{m=1}^{M}h_{j,m,i}s_{m,n-i} \right|^2 \right]                                   +   \sum_{n=0}^{L-2} \sum_{i=n}^{L-2}\mathbb{E}\left[\left|\sum_{m=1}^{M}h_{j,m,L-1+n-i}s_{m,i} \right|^2  \right] + \sum_{n=0}^{L-2}\mathbb{E}\left[\left|z_{j,n}\right|^2\right]\nonumber\\
	=&\sum_{n=0}^{L-2} \sum_{i=0}^{n}\mathbb{E}\left[\frac{1}{N}\left|\sum_{k=0}^{N-1}\sum_{j'=1}^{J}S_{j',m,k}e^{j2\pi (n-i)k/N}\sum_{m=1}^{M}h_{j,m,i}W_{j',m,k} \right|^2 \right]   \nonumber\\
	&+\sum_{n=0}^{L-2} \sum_{i=n}^{L-2}\mathbb{E}\left[\frac{1}{N}\left|\sum_{k=0}^{N-1}\sum_{j'=1}^{J}S'_{j',m,k}e^{j2\pi ik/N}\sum_{m=1}^{M}h_{j,m,L-1+n-i}W_{j',m,k} \right|^2 \right]  + \sum_{n=0}^{L-2}\mathbb{E}\left[\left|z_{j,n}\right|^2\right]  \nonumber\\
	=&\sum_{n=0}^{L-2} \sum_{i=0}^{L-1}\left[\frac{1}{N}\sum_{k=0}^{N-1}\sum_{j'=1}^{J}\left|\sum_{m=1}^{M}h_{j,m,i}W_{j',m,k} \right|^2 \right]  + \sum_{n=0}^{L-2}\mathbb{E}\left[\left|z_{j,n}\right|^2\right]    \nonumber\\
	\leq & \sum_{n=0}^{L-2} \sum_{i=0}^{L-1}\left[\frac{1}{N}\sum_{k=0}^{N-1}\sum_{j'=1}^{J}\sum_{m=1}^{M}\left|h_{j,m,i}\right|^2\left| \sum_{m=1}^MW_{j',m,k} \right|^2 \right]  + \sum_{n=0}^{L-2}\mathbb{E}\left[\left|z_{j,n}\right|^2\right]    \nonumber\\
	\leq & \frac{1}{N}\sum_{n=0}^{L-2} \sum_{i=0}^{L-1}\left[\left|\sum_{m=1}^{M}h_{j,m,i}\right|^2\sum_{k=0}^{N-1}\sum_{j'=1}^{J}\left| \sum_{m=1}^MW_{j',m,k} \right|^2 \right]  + \sum_{n=0}^{L-2}\mathbb{E}\left[\left|z_{j,n}\right|^2\right]    \nonumber\\
	\leq & \frac{M^2}{N}\left[\sum_{n=0}^{L-2} \sum_{i=0}^{L-1}\sum_{m=1}^{M}\left|h_{j,m,i}\right|^2\right]\left[\sum_{k=0}^{N-1}\sum_{j'=1}^{J} \sum_{m=1}^M\left|W_{j',m,k} \right|^2 \right]   + \sum_{n=0}^{L-2}\mathbb{E}\left[\left|z_{j,n}\right|^2\right]   \nonumber\\
	= & M^2(L-1)P_{tx}^{OFDM} \sum_{i=0}^{L-1}\sum_{m=1}^{M}\left|h_{j,m,i}\right|^2  + (L-1)\sigma_{0}^2.
\end{align} 
\end{small}
The energy $E_{OS,j}$ harvested from the OFDM symbols  is expressed as
\begin{small}
\begin{align}
	E_{OS,j}=&\mathbb{E}\left[\sum_{n=L-1}^{N+L-2}\left| r_{j,n} \right|^2\right]\nonumber\\
	=&N\sum_{k=0}^{N-1}\mathbb{E}\left[\left| \sum_{m=1}^{M}H_{j,m,k}S_{m,k} \right|^2\right]+\sum_{k=0}^{N-1}\mathbb{E}\left[\left| Z_{j,k} \right|^2\right]\nonumber\\
	=&N\sum_{k=0}^{N-1}\mathbb{E}\left[\left| \sum_{j'=1}^{J}\sum_{m=1}^{M}H_{j,m,k}W_{j',m,k}S_{j',m,k} \right|^2\right]+\sum_{k=0}^{N-1}\mathbb{E}\left[\left| Z_{j,k} \right|^2\right]\nonumber\\
	=&N\sum_{k=0}^{N-1}\sum_{j'=1}^{J}\mathbf{H}_{j,k}\mathbf{W}_{j',k}\mathbf{W}_{j',k}^\dagger\mathbf{H}_{j,k}^\dagger + N\sigma_0^2  
\end{align}
\end{small}
The total RF power harvested by the $j$-th receiver is formulated as
\begin{small}
\begin{align}
	P_{EH,j}\left( \mathbf{W} , \rho_{j} \right)=&\frac{\rho_j}{N+L-1}\left( E_{CP,j} + E_{OS,j} \right)\nonumber\\
	\leq & \frac{\rho_j}{N+L-1}\left( M^2(L-1)P_{tx}^{OFDM} \sum_{i=0}^{L-1}\sum_{m=1}^{M}\left|h_{j,m,i}\right|^2  + N\sum_{k=0}^{N-1}\sum_{j'=1}^{J}\mathbf{H}_{j,k}\mathbf{W}_{j',k}\mathbf{W}_{j',k}^\dagger\mathbf{H}_{j,k}^\dagger \right) +\rho_j\sigma_{0}^2\nonumber\\
	=& \widetilde{P}_{EH,j}\left( \mathbf{W} , \rho_{j} \right),
\end{align}
\end{small}
where $\widetilde{P}_{EH,j}\left( \mathbf{W} , \rho_{j} \right)$ is the upper bound of  $P_{EH,j}\left( \mathbf{W} , \rho_{j} \right)$. Given an input power of $\widetilde{P}_{EH,j}(\mathbf{W},\rho_{j})$, the output DC is expressed as
\begin{align}
	i_j^{MO}(\mathbf{W},\rho_{j}) \approx k_0 + k_1 \widetilde{P}_{EH,j}(\mathbf{W},\rho_{j}) + k_2 \widetilde{P}^2_{EH,j}(\mathbf{W},\rho_{j}).
\end{align}

\subsection{Fair-Throughput Maximisation}
The fair-throughput maximisation problem of the MISO-OFDM is formulated as
\begin{align}
&\text{(P5):}\max_{ \left\{\mathbf{W}_{j,k}\right\},\left\{\rho_j\right\} }R_{fair}^{MO}\label{Problem 5}\\
&\text{s. t.    } R_j\left( \mathbf{W} , \rho_{j} \right) \geq R_{fair}^{MO},\forall j=1,\cdots,J,\tag{\ref{Problem 5}a}\label{{Problem5}a}\\
&\text{         }\ \ \ \ \  i_j^{MO}(\mathbf{W},\rho_{j})  \geq \mathcal{I}_j^{MO},\forall j=1,\cdots,J,\tag{\ref{Problem 5}b}\label{{Problem5}b}\nonumber\\
&\text{         }\ \ \ \ \  \Vert \mathbf{W}_{j,k} \Vert_2^2\leq p_{j,k},\forall j,k,\tag{\ref{Problem 5}c}\label{{Problem5}c}\nonumber\\
&\text{         }\ \ \ \ \  0\leq\sum_{m=1}^{M}\sum_{k=0}^{N-1}p_{j,k}\leq NP_{tx}^{OFDM},\tag{\ref{Problem 5}d}\label{{Problem5}d}\\
&\text{         }\ \ \ \ \  0\leq\rho_j\leq1,\tag{\ref{Problem 5}e}\label{{Problem5}e}
\end{align}
where the objective function  $R_{fair}^{MO}$ is the fair-throughput and  $\mathcal{I}_j^{MO}$ is the DC requirement of the $j$-th receiver. \eqref{{Problem5}a} represents the throughput constraint for each receiver. \eqref{{Problem5}b} represents the power requirement for all the receivers. \eqref{{Problem5}c} represents the power constraint for the beamformers and \eqref{{Problem5}d} represents the total transmit power constraint for the transmitter. Finally, \eqref{{Problem5}e} represents the power splitter constraint. Since $\gamma_{j,k}\left( \mathbf{W} , \rho_{j} \right)$ is a fractional function and $\widetilde{P}_{EH,j}\left( \mathbf{W} , \rho_{j} \right)$ is a quadratic function,  \eqref{{Problem5}a} and \eqref{{Problem5}b} are non-convex constraints and (P5) is a non-convex problem. By introducing the auxiliary  variable $\{\psi_{j,k}^{ID}\}$, the throughput $R_j$ is reformulated as
\begin{small}
\begin{align}
	\widehat{R}_j\left( \mathbf{W} , \rho_{j} , \boldsymbol{\psi}_{j}^{ID} \right)&=\frac{B}{N+L-1}\log_2\left[ 1+2\sqrt{N(1-\rho_j)}\mathcal{R}\left( \psi_{j,k}^{ID\dagger}\mathbf{H}_{j,k}\mathbf{W}_{j,k} \right) \nonumber\right.\\ 
	&\left.  - \left| \psi_{j,k}^{ID} \right|^2\left( N(1-\rho_j)\sum_{j'\neq j}\left| \mathbf{H}_{j,k}\mathbf{W}_{j,k}\right|^2 +(1-\rho_j)\sigma_0^2 + \sigma_{cov}^2 \right)             \right].
\end{align}
\end{small}
where $\boldsymbol{\psi}_j^{ID}=\left[ \psi_{j,0}^{ID}, \cdots, \psi_{j,N-1}^{ID} \right]^T$.
Similarly, after adopting the quadratic transformation, the constraint \eqref{{Problem5}b} is reformulated as
\begin{align}
\widetilde{P}_{EH,j}(\mathbf{W},\rho_{j})\geq \frac{-k_1+\sqrt{k_1^2+4k_2\mathcal{I}_j^{MO}}}{2k_2},\forall j=1,\cdots,N.\label{powerconstraint2}
\end{align}
By introducing the auxiliary variable $\{ \psi_{j,k}^{EH} \}$, the average energy harvesting power $\widetilde{P}_{EH,j}(\mathbf{W},\rho_{j})$ is reformulated as
\begin{small}
\begin{align}
		\widehat{P}_{EH,j}\left( \mathbf{W} , \rho_{j} , \boldsymbol{\psi}_j^{EH} \right)=&\frac{4\rho_jP_{tx}(L-1)}{N+L-1}  \sum_{i=0}^{L-1}\sum_{m=1}^{M}\left|h_{j,m,i}\right|^2\nonumber\\
	&+\frac{\rho_jN}{N+L-1}\sum_{k=0}^{N-1}\sum_{j'=1}^{J}\left[ 2\mathcal{R}\left(\psi_{j,k}^{EH\dagger}\mathbf{H}_{j,k}\mathbf{W}_{j,k}\right)-\left| \psi_{j,k}^{EH} \right|^2 \right] + \rho_j\sigma_{0}^2,
\end{align}  
\end{small}
where $\boldsymbol{\psi}_j^{EH}=\left[ \psi_{j,0},\cdots,\psi_{j,N-1} \right]^T$.
Then (P5) can be reformulated as
\begin{align}
&\text{(P6):}\max_{ \left\{\mathbf{W}_{j,k}\right\},\left\{\rho_j\right\},\{\psi_{j,k}^{ID}\},\{\psi_{j,k}^{EH}\} }\widehat{R}_{fair}^{MO}\label{Problem 6}\\
&\text{s. t.    } \widehat{R}_j\left( \mathbf{W} , \rho_{j} , \boldsymbol{\psi}_j^{ID} \right) \geq \widehat{R}_{fair}^{MO},\forall j=1,\cdots,J,\tag{\ref{Problem 6}a}\label{{Problem6}a}\\
&\text{         }\ \ \ \ \  \widehat{P}_{EH,j}\left( \mathbf{W} , \rho_{j} , \boldsymbol{\psi}_j \right)\geq \frac{-k_1+\sqrt{k_1^2+4k_2\mathcal{I}_j^{MO}}}{2k_2},\forall j=1,\cdots,J,\tag{\ref{Problem 6}b}\label{{Problem6}b}\nonumber\\
&\text{         }\ \ \ \ \  \eqref{{Problem5}c},\eqref{{Problem5}d},\eqref{{Problem5}e}\nonumber.
\end{align}
In order to solve  (P6), we need to initialize $\{\psi_{j,k}^{EH}\}$ and $\{\mathbf{W}_{j,k}\}$ for satisfying the constraint \eqref{{Problem6}b}. Moreover, the receive power maximisation problem is formulated as
\begin{align}
&\text{(P7):}\max_{ \left\{\mathbf{W}_{j,k}\right\},\{\psi_{j,k}^{EH}\} }P_{fair}\label{Problem 7}\\
&\text{         }\ \ \ \ \  \eqref{{Problem6}b},  \eqref{{Problem5}c},\eqref{{Problem5}d}\nonumber,
\end{align}
where $P_{fair}$ is the fair-power. By alternatively  optimising  $\{\mathbf{W}_{j,k}\}$ and $\{\psi_{j,k}^{EH}\}$, we can solve (P7) and obtain a feasible solution. Firstly, we fix $\psi_{j,k}^{EH}$ and (P7) can be reformulated as
\begin{align}
&\text{(P7-1):}\max_{ \left\{\mathbf{W}_{j,k}\right\} }P_{fair}\label{Problem 7-1}\\
&\text{         }\ \ \ \ \  \eqref{{Problem6}b},  \eqref{{Problem5}c},\eqref{{Problem5}d}\nonumber.
\end{align}
(P7-1) is a convex problem and we can solve it by any CVX tool box. By fixing $\{\mathbf{W}_{j,k}\}$, (P7) can be reformulated as
\begin{align}
&\text{(P7-2):}\max_{ \left\{\psi_{j,k}^{EH}\right\} }P_{fair}\label{Problem 7-2}\\
&\text{         }\ \ \ \ \  \eqref{{Problem6}b}.
\end{align}
The close form solution of (P7-2) can be expressed as
\begin{align}
	\psi_{j,k}^{EH*}=&\mathbf{H}_{j,k}\mathbf{W}_{j,k}.\label{eq:49}
\end{align}
The alternative optimisation (AO)  based method algorithm for solving (P7) is detailed in Algorithm 3.
\begin{algorithm}[!t]
	\linespread{1}\selectfont
	\caption{AO based algorithm for solving (P7).}
	\label{alg:alg3}
	\footnotesize
	\begin{algorithmic}[1]
		\REQUIRE\
		Channel fading coefficient of $\{\mathbf{h}_{j,m}\}$; Transmit power of the BS $P_{tx}^{OFDM}$;  DC requirement $\mathcal{I}_j^{MO}$; Error tolerance $\epsilon$.
		\ENSURE\
		Feasible solution $\psi_{j,k}^{EH*}$.
		\STATE Initialize $\{\rho_{j}\}\leftarrow 1$,  $\{\Psi_{j,k}^{EH}\}\leftarrow 1$, $P_{fair}^1\leftarrow \epsilon$, $P_{fair}^2\leftarrow -\epsilon$;
		\WHILE{$|P_{fair}^1-P_{fair}^2|\geq\epsilon$}
		\STATE $P_{fair}^2\leftarrow P_{fair}^1$;
		\STATE Update $\left\{ \mathbf{W}_{j,k}^* \right\}$ by solving (P7-1);
		\STATE Update ${\psi}_{m,k}^{EH*}$ by Eq. \eqref{eq:49};
		\STATE Update $P_{fair}^1\leftarrow P_{fair}$;
		\ENDWHILE
		\RETURN $\{\psi_{j,k}^{EH*}\}$.
	\end{algorithmic}
\end{algorithm}

Then we solve (P6) by alternatively optimising $\{\mathbf{W}_{j,k}\}$, $\{\rho_j\}$, $\{\psi_{j,k}^{ID}\}$ and $\{\psi_{j,k}^{EH}\}$. First, by fixing $\{\rho_j\}$, $\{\psi_{j,k}^{ID}\}$ and $\{\psi_{j,k}^{EH}\}$, (P6) is reformulated as 
\begin{align}
&\text{(P6-1):}\max_{ \left\{\mathbf{W}_{j,k}\right\} }\widehat{R}_{fair}^{MO}\label{Problem 6-1}\\
&\text{s. t.    } \eqref{{Problem6}a},\eqref{{Problem6}b}, \eqref{{Problem5}c},\eqref{{Problem5}d}\nonumber.
\end{align}
Moreover, when we fix $\{\mathbf{W}_{j,k}\}$, $\{\psi_{j,k}^{ID}\}$ and $\{\psi_{j,k}^{EH}\}$, (P6) is reformulated as
\begin{align}
&\text{(P6-2):}\max_{\left\{\rho_j\right\} }\widehat{R}_{fair}^{MO}\label{Problem 6-2}\\
&\text{s. t.    } \eqref{{Problem6}a},\eqref{{Problem6}b},  
\eqref{{Problem5}e}\nonumber.
\end{align}
Since (P6-1) and (P6-2) are both the convex problem, we can solve them by the convex optimisation tool. Finally, by fixing $\{\mathbf{W}_{j,k}\}$ and $\{\rho_j\}$, (P6) is reformulated as
\begin{align}
&\text{(P6-3):}\max_{ \{\psi_{j,k}^{ID}\},\{\psi_{j,k}^{EH}\} }\widehat{R}_{fair}^{MO}\label{Problem 6-3}\\
&\text{s. t.    } \eqref{{Problem6}a},\eqref{{Problem6}b}.\nonumber
\end{align}
(P4-3) is a quadratically constrained quadratic programs (QCQP) and we can obtain the colsed-form solutions. The optimal $\{\psi_{j,k}^{EH*}\}$ is expressed as  \eqref{eq:49} and the optimal $\{\psi_{j,k}^{ID*}\}$ is expressed as
\begin{align}
	\psi_{j,k}^{ID*}=&\frac{\sqrt{N(1-\rho_j)}\mathbf{H}_{j,k}\mathbf{W}_{j,k} }{\left( N(1-\rho_j)\sum_{j'\neq j}\left|\mathbf{H}_{j,k}\mathbf{W}_{j',k}\right|^2 +(1-\rho_j)\sigma_0^2 + \sigma_{cov}^2 \right) }.\label{eq:53}
\end{align}
The FP based alternating optimisation for solving (P4) is detailed  in Algorithm 3. The proof of the convergence for the Algorithm 4 is similar to the Algorithm 1. However, since the constraint \eqref{{Problem5}c} is non-convex, we can only prove (P5) converges to a stationary point.
\begin{algorithm}[!t]
	\linespread{1}\selectfont
	\caption{FP based alternating optimisation for solving (P5).}
	\label{alg:alg4}
	\footnotesize
	\begin{algorithmic}[1]
		\REQUIRE\
		Channel fading coefficient of $\mathbf{h}_{j,m}$; Transmit power of the BS $P_{tx}^{OFDM}$; DC requirement $\mathcal{I}_j^{MO}$; Feasible auxiliary variable $\{ \psi_{j,k}^{EH} \}$; Error tolerance $\epsilon$.
		\ENSURE\
		Optimal fair-throughput $R_{fair}^{MO*}$.
		\STATE Initialize $\{\rho_{j}\}\leftarrowtail 1$, $\{\psi_{j,k}^{ID}\}\leftarrow 1$, $\widehat{R}_{fair}^{MO1}\leftarrow\epsilon$, $\widehat{R}_{fair}^{MO2}\leftarrow-\epsilon$
		\WHILE{$|\widehat{R}_{fair}^{MO1}-\widehat{R}_{fair}^{MO2}|\geq\epsilon$}
		\STATE $\widehat{R}_{fair}^{MO2}\leftarrow \widehat{R}_{fair}^{MO1}$;
		\STATE Update $\left\{ \mathbf{W}_{j,k}^* \right\}$ by solving (P5-1);
		\STATE Update $\left\{\rho_j^*\right\}$ by solving (P5-2);
		\STATE Update ${\psi}_{j,k}^{EH*}$ by Eq. \eqref{eq:49};
		\STATE Update ${\psi}_{j,k}^{ID*}$ by Eq. \eqref{eq:53};
		\STATE Update $\widehat{R}_{fair}^{MO1}\leftarrow \widehat{R}_{fair}^{MO}$;
		\ENDWHILE
		\STATE Update $R_{fair}^{MO*} \leftarrow \widehat{R}_{fair}^{MO*}$
		\RETURN $R_{fair}^{MO*}$.
	\end{algorithmic}
\end{algorithm}

\subsection{Sum-Throughput Maximisation}
The sum-throughput maximisation problem can be formulated as
\begin{align}
&\text{(P8):}\max_{ \left\{\mathbf{W}_{j,k}\right\},\left\{\rho_j\right\} }R_{sum}^{MO}=\sum_{j=1}^{J}R_j\left( \mathbf{W} , \rho_{j} \right)\label{Problem 8}\\
&\text{s. t.    } \eqref{{Problem5}b},\eqref{{Problem5}c},\eqref{{Problem5}d},\eqref{{Problem5}e}.
\end{align}
Similarly, by introducing the auxilary variables $\{ \psi_{j,k}^{EH} \}$ and $\{ \psi_{j,k}{ID} \}$, (P8) is reformulated as
\begin{align}
&\text{(P9):}\max_{ \left\{\mathbf{W}_{j,k}\right\},\left\{\rho_j\right\}, \{ \psi_{j,k}^{ID} \}, \{\psi_{j,k}^{EH}\} }\widehat{R}_{sum}^{MO}=\sum_{j=1}^{J}\widehat{R}_j\left( \mathbf{W} , \rho_{j} , \boldsymbol{\psi}_j \right)\label{Problem 9}\\
&\text{s. t.    } \eqref{{Problem6}a},\eqref{{Problem5}c},\eqref{{Problem5}d},\eqref{{Problem5}e}.
\end{align}
We are capable of solving (P9) by alternatively optimising $\{\mathbf{W}_{j,k}\}$, $\{\rho_{j}\}$, $\{ \psi_{j,k}^{ID} \}$ and $\{\psi_{j,k}^{EH}\}$. First, by giving fixed $\{\rho_{j}\}$, $\{ \psi_{j,k}^{ID} \}$ and $\{\psi_{j,k}^{EH}\}$, (P9) is reformulated as
\begin{align}
&\text{(P9-1):}\max_{ \left\{\mathbf{W}_{j,k}\right\} }\widehat{R}_{sum}^{MO}=\sum_{j=1}^{J}\widehat{R}_j\left( \mathbf{W} , \rho_{j} , \boldsymbol{\psi}_j \right)\label{Problem 9-1}\\
&\text{s. t.    } \eqref{{Problem6}a},\eqref{{Problem5}c},\eqref{{Problem5}d}.
\end{align}
Given fixed $\{\mathbf{W}_{j,k}\}$, $\{ \psi_{j,k}^{ID} \}$ and $\{\psi_{j,k}^{EH}\}$, (P9) is reformulated as
\begin{align}
&\text{(P9-2):}\max_{ \left\{\rho_j\right\} }\widehat{R}_{sum}^{MO}=\sum_{j=1}^{J}\widehat{R}_j\left( \mathbf{W} , \rho_{j} , \boldsymbol{\psi}_j \right)\label{Problem 9-2}\\
&\text{s. t.    } \eqref{{Problem6}a},\eqref{{Problem5}e}.
\end{align}
The problems (P9-1) and (P9-2) are convex, we can solve them by any CVX tool box. Finally, given fixed $\{\mathbf{W}_{j,k}\}$ and $\{\rho_{j}\}$, (P9) can be reformulated as 
\begin{align}
&\text{(P9-3):}\max_{  \{ \psi_{j,k}^{ID} \}, \{\psi_{j,k}^{EH}\} }\widehat{R}_{sum}^{MO}=\sum_{j=1}^{J}\widehat{R}_j\left( \mathbf{W} , \rho_{j} , \boldsymbol{\psi}_j \right)\label{Problem 9-3}\\
&\text{s. t.    } \eqref{{Problem6}a}.
\end{align}
By letting $ \partial \widehat{R}_{sum}^{MO} / \partial \psi_{j,k}^{EH} = 0 $ and $ \partial \widehat{R}_{sum}^{MO} / \partial \psi_{j,k}^{ID} = 0 $, the optimal $\{ \psi_{j,k}^{EH*} \}$ and $\{ \psi_{j,k}^{ID*} \}$ can be obtained which are expressed as \eqref{eq:49} and \eqref{eq:53}, respectively. The FP based alternating optimising (P8) is detailed in Algorithm 5. The proof of the convergence for the Algorithm 5 is similar to the Algorithm 1. However, since the constraint \eqref{{Problem5}c} is non-convex, we can only prove (P8) converges to a stationary point.
\begin{algorithm}[!t]
	\linespread{1}\selectfont
	\caption{FP based alternating optimisation for solving (P8).}
	\label{alg:alg5}
	\footnotesize
	\begin{algorithmic}[1]
		\REQUIRE\
		Channel fading coefficient of $\mathbf{h}_{j,m}$; Transmit power of the BS $P_{tx}^{OFDM}$; DC requirement $\mathcal{I}_j^{MO}$; Feasible auxiliary variable $\{ \psi_{j,k}^{EH} \}$; Error tolerance $\epsilon$.
		\ENSURE\
		Optimal fair-throughput $R_{fair}^{MO*}$.
		\STATE Initialize $\{\rho_{j}\}\leftarrowtail 1$, $\{\psi_{j,k}^{ID}\}\leftarrow 1$, $\widehat{R}_{sum}^{MO1}\leftarrow\epsilon$, $\widehat{R}_{sum}^{MO2}\leftarrow\epsilon$
		\WHILE{$|\widehat{R}_{sum}^{MO1}-\widehat{R}_{sum}^{MO2}|\geq\epsilon$}
		\STATE $\widehat{R}_{sum}^{MO2}\leftarrow \widehat{R}_{sum}^{MO1}$;
		\STATE Update $\left\{ \mathbf{W}_{j,k}^* \right\}$ by solving (P9-1);
		\STATE Update $\left\{\rho_j^*\right\}$ by solving (P9-2);
		\STATE Update ${\psi}_{j,k}^{EH*}$ by Eq. \eqref{eq:49};
		\STATE Update ${\psi}_{j,k}^{ID*}$ by Eq. \eqref{eq:53};
		\STATE Update $\widehat{R}_{sum}^{MO1}\leftarrow \widehat{R}_{fair}^{MO}$;
		\ENDWHILE
		\STATE Update $R_{sum}^{MO*} \leftarrow \widehat{R}_{sum}^{MO}$
		\RETURN $R_{fair}^{MO*}$.
	\end{algorithmic}
\end{algorithm}

\begin{appendices}

\section{Proof for Theorem 1}
The symbol $s_n$ can be expressed as
\begin{align}
	s_n=\frac{1}{\sqrt{N}}\sum_{k=0}^{N-1}S_k e^{j2\pi nk/N}\label{eq27}. 
\end{align}
Observe from \eqref{eq27}, the symbols $S_0,\cdots,S_{N-1}$ are the CSGR variables and independent each other. Therefore, the symbol $s_n$ is the CSGR and the expection and the variance can be expression can be expression as 
\begin{align}
	\mathbb{E}\left[s_n\right]&=\frac{1}{\sqrt{N}}\sum_{k=0}^{N-1}\mathbb{E}\left[S_k\right]e^{j2\pi nk/N}=0,\\
	Var\left[s_n\right]&=\mathbb{E}\left[s_ns_n^\dagger\right]\nonumber\\
	&=\frac{1}{{N}}\sum_{k=0}^{N-1}\mathbb{E}\left[S_ke^{j2\pi nk/N}S_k^\dagger e^{-j2\pi nk/N}\right]=\frac{1}{{N}}\sum_{k=0}^{N-1}p_{k}=\frac{1}{{N}}\sum_{k=0}^{N-1}\sum_{m=0}^{M-1}p_{m,k}=P_{tx}.
\end{align}
The IDFT of the OFDM symbol vector $\mathbf{S}$ is $\mathbf{s}=\mathbf{S}\mathbf{F}^{\dagger}$, where the DFT matrix $\mathbf{F}$ is a unitary matrix and its different rows or columns are orthogonal. Therefore, for $\forall n_1\neq n_2$, the symbols $s_{n_{1}}$ and $s_{n_2}$  are independent each other according to the orthogonal theorem \cite{Probability}. 
Finally, we proof that $\mathbf{s}\sim \mathcal{CN}\left(0,P_{tx}\mathbf{I}\right)$.
\section{Proof for Theorem 2}
According to the Theorem 1, we konw that $s_{n_1}\sim \mathcal{CN}\left(0,P_tx\right)$ and $s_{n_2}'\sim \mathcal{CN}\left(0,P_tx\right)$. and they can expressed as
\begin{align}
	s_{n_1}&=\frac{1}{\sqrt{N}}\sum_{k=0}^{N-1}S_k e^{j2\pi n_1 k/N},\\
	s_{n_2}'&=\frac{1}{\sqrt{N}}\sum_{k=0}^{N-1}S_k' e^{j2\pi n_2 k/N}.
\end{align}
The covariance of the $s_{n_1}$ and $s_{n_2}'$ can be expressed as
\begin{align}
	Cov\left[ s_{n_1},s_{n_2}' \right]&=\mathbb{E}\left[ s_{n_1}s_{n_2}'^\dagger \right]\nonumber\\
	&=\frac{1}{N}\mathbb{E}\left[ \sum_{k=0}^{N-1}S_k e^{j2\pi n_1 k/N}  \sum_{k=0}^{N-1}S_k'^\dagger e^{-j2\pi n_2 k/N} \right]\nonumber\\
	&=\frac{1}{N}\mathbb{E}\left[ \sum_{k_1=0}^{N-1}\sum_{k_2=0}^{N-1}S_{k_1}S_{k_2}'^\dagger e^{j2\pi\left(n_1k_1-n_2k_2\right)/N} \right]\nonumber\\
	&=\frac{1}{N} \sum_{k_1=0}^{N-1}\sum_{k_2=0}^{N-1}\mathbb{E}\left[S_{k_1}S_{k_2}'^\dagger\right] e^{j2\pi\left(n_1k_1-n_2k_2\right)/N} \nonumber\\
	&=0.
\end{align}
Therefore, the symbols in the different OFDM period are uncorrelated. If two Guassian variables are uncorrelated, then they are independent each other \cite{Probability}. Finally, the symbols $s_{n_1}$ and $s_{n_2}'$ are independent.

\section{Proof for the Convergence of Algorithm 1}
We denote the $k$-th iteration of the SISO-OFDM fair-throughput in the Algorithm 1 as  $R_{fair}^{SO(k)}$, while the variables are denoted as $\mathbf{p}^{(k)}$, $\rho_{j}^{(k)}$ and $\boldsymbol{\psi}_{j}^{(k)}$, respectively. Then the monotonicity of $R_{fair}^{SO(k)}$ with respect to $k$ is derived as
\begin{small}
\begin{align}
	R_{fair}^{SO(k)}=&\min_j R_{j} \left( \mathbf{p}^{(k)},\rho_{j}^{(k)} \right)\nonumber\\
	\overset{(a)}{=}&\min_j \widetilde{R}_j \left( \mathbf{p}^{(k)},\rho_{j}^{(k)},\boldsymbol{\psi}_{j}^{(k)} \right)\nonumber\\
	\leq &\min_j \widetilde{R}_j \left( \mathbf{p}^{(k+1)},\rho_{j}^{(k)},\boldsymbol{\psi}_{j}^{(k)} \right)\nonumber\\
	\leq &\min_j \widetilde{R}_j \left( \mathbf{p}^{(k+1)},\rho_{j}^{(k+1)},\boldsymbol{\psi}_{j}^{(k)} \right)\nonumber\\
	\leq &\min_j \widetilde{R}_j \left( \mathbf{p}^{(k+1)},\rho_{j}^{(k+1)},\boldsymbol{\psi}_{j}^{(k+1)} \right) \nonumber\\
	=& \min_j R_{j} \left( \mathbf{p}^{(k+1)},\rho_{j}^{(k+1)} \right)\nonumber\\
	\overset{(b)}{=}& R_{fair}^{SO(k+1)}.
\end{align}
\end{small}
(a) and (b) hold since (P1) and (P2)  have the same optimal objective value when the auxiliary variable $\boldsymbol{\psi}_{j}$ is optimal. Then we prove $R_{fair}^{SO}$ is bounded by a constant, which is derived as
\begin{small}
\begin{align}
	R_{fair}^{SO}=&\frac{B}{N+L-1} \log_2\left( 1 +  \frac {(1-\rho_j)N\Vert H_{j,k}\Vert_2^2p_{j,k}}   {(1-\rho_j)\sum_{j'\neq j}N\Vert H_{j',k}\Vert_2^2p_{j',k}  + (1-\rho_j)\sigma_{0}^2 + \sigma_{cov}^2 }    \right)  \nonumber\\
	\overset{(c)}{\leq} & \frac{B}{N+L-1} \log_2\left( 1 +  \frac {(1-\rho_j)N\Vert H_{j,k}\Vert_2^2p_{j,k}}   {  \sigma_{cov}^2 }    \right)  \nonumber\\
	\overset{(d)}{\leq} & \frac{B}{N+L-1} \log_2\left( 1 +  \frac {N^2\Vert H_{j,k}\Vert_2^2P_{tx}^{OFDM}}   {  \sigma_{cov}^2 }    \right)  \nonumber\\
	\leq& +\infty
\end{align}
\end{small}
where (c) holds, since we remove the interference and the antenna noise in the denominator of the SINR; (d) holds given the constraints \eqref{{Problem1}c} and \eqref{{Problem1}d}. Finally, $R_{fair}^{SO}$ increases monotonously and it is upper-bounded by a finite value. Therefore, Algorithm 1 converges.

\end{appendices}

\end{document}